\documentclass[letterpaper, 10 pt, conference]{ieeeconf}  

\usepackage{graphicx}
\usepackage{amsmath} 
\usepackage{amssymb}  

\usepackage{amsthm}
\newtheorem{lemm}{Lemma}
\newtheorem{prop}{Proposition}
\theoremstyle{definition}
\newtheorem{deff}{Definition}
\newtheorem{remark}{Remark}
\newtheorem{example}{Example}

\title{\LARGE \bf
Necessary and Sufficient Conditions for the Optimization-Based Concurrent Execution of Learned Robotic Tasks
}

\author{Sheikh A. Tahmid and Gennaro Notomista 
\thanks{Sheikh A. Tahmid is with the Department of Electrical and Computer Engineering, University of Waterloo, Waterloo, ON, Canada.
	{\tt\small sheikh.abrar.tahmid@uwaterloo.ca}}%
\thanks{Gennaro Notomista is with the Department of Electrical and Computer Engineering, University of Waterloo, Waterloo, ON, Canada.
	{\tt\small gennaro.notomista@uwaterloo.ca}}%
}

\begin{document}

\maketitle
\thispagestyle{empty}
\pagestyle{empty}

\begin{abstract}
In this work, we consider the problem of executing multiple tasks encoded by value functions, each learned through Reinforcement Learning, using an optimization-based framework.
Prior works develop this framework but did not address when learned value functions can be concurrently executed.
This work's main contributions consist of theorems which provide necessary and sufficient conditions to
concurrently execute sets of learned tasks within subsets of the state space using the previously proposed
min-norm controller.
These theorems provide insight into when learned control tasks can be made concurrently executable, when they may already be so, and when concurrent execution is not possible under the proposed framework.
We also extend the proposed framework to account for value functions trained with a discount factor, making it more compatible with standard RL practices.
\end{abstract}

\section{INTRODUCTION}
In recent years, interest has grown from the robotics, control systems, and machine learning communities in applying Reinforcement Learning (RL) to robotic control tasks.
RL is deeply connected to the field of optimal control \cite{bertsekas}.
If a robotic task can be posed as an optimal control problem, its solution, given enough
time and computational resources, may be approximated by an RL algorithm to produce some learned policy.
RL algorithms can potentially solve optimal control problems that are otherwise intractable using control-theoretic or numerical optimization-based methods.
RL algorithms make very general assumptions, to varying degrees, on the cost to optimize or the model of the system, giving the field of RL the potential to tackle difficult control problems in robotics.
In both optimal control and RL, the solved or learned task is often encoded in a value or cost-to-go function -- terms which we use interchangeably in this paper.
\par
For robotic systems that exhibit redundancy -- the physical capability of concurrently executing multiple tasks,
such as multi-robot systems \cite{redundancyBook}, \cite{multiRobotRedundancy1} and manipulators \cite{redundancyBook}, \cite{robotManipulatorsRedundancy1} -- 
it can be useful to exploit that redundancy to concurrently execute multiple control tasks at the same time.
Previous works have developed relevant theory for robotic systems concurrently executing explicitly designed tasks such as Jacobian-based \cite{jacobianTasks} and Extended Set-Based (ESB) \cite{esb} tasks.
However, due to the potential of RL algorithms, as discussed previously, this work is concerned with the concurrent execution of multiple control tasks learned using RL.
\par
Previous works in \cite{gennaroDars} and \cite{meAamas} develop a framework for simultaneously executing multiple tasks, where each task is learned using an RL algorithm and each resulting value function is executed through a constraint in a pointwise min-norm controller.
The benefit of this overall approach, initially proposed in \cite{gennaroDars} and further expanded on in \cite{meAamas}, is that it provides a method of combining multiple learned tasks where the subsets of tasks and priorities between
them may vary over time.
The work in \cite{meAamas} focuses on training tasks so that the learned value functions satisfy a notion of \textit{independence}, allowing trained tasks to be combined and executed together that may have otherwise had conflicting objectives.\par
The main contribution of this paper is the specification of conditions for when a set of learned value functions can be combined and executed concurrently with the proposed min-norm controller. 
We present propositions which provide necessary and sufficient conditions for a set of learned value functions to be \textit{concurrently executable} -- a term we define along with others in this work.
While the prior works \cite{gennaroDars}, \cite{meAamas} develop the mentioned framework, they leave unanswered the fundamental question of when learned value functions can be concurrently executed using the proposed min-norm controller.
The theory presented in this paper answers this question, providing insight into when learned control tasks are possible to be made concurrently executable, when they might already inherently be concurrently executable and when it is not possible at all to make a set of learned tasks concurrently executable.
Beyond its immediate application to the framework of \cite{gennaroDars}, \cite{meAamas}, the theorems developed in this work may provide a foundation for the rigorous analysis or development of multi-objective RL methods.
\par
\subsection{Related Works}
Many frameworks for exploiting the redundancy of robotic systems in order to concurrently execute multiple tasks have been proposed over the years (see \cite{khatib}, \cite{robotManipulatorsRedundancy1}, \cite{esb} to name a few).
There also exist works which not only propose frameworks, but analyze the execution of such algorithms.
The work in \cite{jacobianTasks} analyzes the stability of algorithms for executing multiple Jacobian-based tasks for redundant, kinematic manipulators, utilizing definitions of \textit{independence} and \textit{orthogonality}.
The work in \cite{esb} uses Extended Set-Based (ESB) Tasks to generalize Jacobian-based tasks and analyze the concurrent execution of ESB tasks, also generalizing notions of \textit{independence} and \textit{orthogonality} for ESB tasks.\par
A model-based framework for the execution of multiple robotic tasks, with each task learned through deep RL, was proposed in \cite{gennaroDars}.
With certain assumptions on the system and reward function used for each task, the work in \cite{gennaroDars} proposes a pointwise min-norm controller where each task's learned value function is encoded as a constraint of the quadratic program. An expression derived from a modified version of Sontag's formula, described in the work, is used to recover the optimal input for each task. Slack variables, based on previous work in \cite{gennaroEcc}, are used to prioritize the execution of different learned tasks.\par
Subsequent work in \cite{meAamas} extends the work in \cite{gennaroDars}. The work in \cite{meAamas} defines notions of \textit{independence} and \textit{orthogonality} based on the notions of the same name used in \cite{jacobianTasks} and \cite{esb} to analyze the relationships between learned value functions and the ability to simultaneously execute multiple such learned tasks together. The work in \cite{meAamas} proposes a cost on the input which can be used to make learned value functions \textit{independent} to previously trained value functions, allowing one to train a sequence of \textit{independent} value functions that are then possible to place as constraints in the min-norm controller proposed in \cite{gennaroDars} and concurrently execute the learned tasks together. Additionally, a variant of value iteration for continuous systems is proposed to approximate value functions with the proposed input cost which can be used in addition to existing RL algorithms.
\subsection{Contributions}
\begin{itemize}
\item{
The main contributions of this paper are in Section IV where we present propositions which provide necessary and sufficient conditions for a set of learned value functions to be \textit{concurrently executable}. \textit{Concurrent executability} and \textit{concurrent controllability} are terms which we define in Section III, similar to definitions of \textit{independence} and \textit{orthogonality} in \cite{jacobianTasks}, \cite{esb} and \cite{meAamas}.}
\item{
To also further develop the framework developed in \cite{gennaroDars} and \cite{meAamas}, we also present in Subsection IV.C a modified expression for the right-hand-side of the constraints of the min-norm controller to account for value functions trained with a discount factor -- something not done in the work in \cite{gennaroDars} that introduced the min-norm controller.
}
\end{itemize}


\section{MATHEMATICAL BACKGROUND}
In this section, we first state the assumptions we make in our chosen problem setting.
We then formally describe the min-norm controller and optimization-based framework for executing learned tasks developed in \cite{gennaroDars}, \cite{meAamas}.
\subsection{Assumptions on System and Learned Value Functions}
Firstly, we assume that the system is deterministic and control-affine. The dynamics can be represented with: 
$$
\dot{x}(t) = f(x) + g(x) u \eqno{(1)}
$$
where $x \in \mathcal{X} \subseteq \mathbb{R}^n$ and $u \in \mathcal{U} \subseteq \mathbb{R}^m$ denote the state and input, respectively, and 
$f : \mathbb{R}^n \rightarrow \mathbb{R}^n$ and $g : \mathbb{R}^n \rightarrow \mathbb{R}^{n \times m}$ are continuously differentiable vector fields.\par
We consider the case where we have $N$ tasks which we want to concurrently execute.
We assume that each task $i \in \{1, \ldots, N\}$ has an associated state cost, $q_i(x)$, that is nonnegative.
To learn each task $i$, we consider solving the following infinite horizon optimal control problem:
$$
J^{\star}_i(x(t)) = \min_{u(\cdot)} \int_t^{\infty} e^{- \beta \left(\tau - t \right) } \left( q_i(x) + u^{\top} R_i(x) u \right) d \tau \eqno{(2)}
$$
where $\beta \ge 0$ and $R_i : \mathbb{R}^n \rightarrow \mathbb{S}^{n}_{++} $ maps each state $x \in \mathcal{X}$ to a positive definite matrix. 
We assume that each task $i$ has at least one reachable terminal goal state $x_T \in \mathcal{X}$ where $q_i(x_T) = 0$
and $x_T$ is an equilibrium point.
Since the cost-to-go function in (2) cannot be solved analytically in the general case, we opt to use deep RL algorithms to learn approximations which we will denote as $\tilde{J}_i$ for each task $i$.
Note that RL is typically posed as a maximization problem, where a function such as one in (2) would be referred to as a value function, but in this paper, we pose it as a minimization problem. Actor-critic methods \cite{sac}, \cite{actorCriticSurvey} which learn both a policy and value function are examples of RL algorithms which can be used.
The work in \cite{meAamas} proposes a variant of fitted value iteration to specifically approximate a cost-to-go function of the form in (2), when the dynamics are control-affine, building upon work in \cite{cfvi}.

\subsection{The Pointwise Min-Norm Controller}
Given a set of learned cost-to-go functions $\tilde{J}_1,\ldots,\tilde{J}_N$, each implemented with a neural network, we can select inputs that execute each task concurrently using the pointwise min-norm controller proposed in \cite{gennaroDars}:
\begin{equation} \tag{3}
\begin{aligned}
	\min_{u \in \mathcal{U}, \delta \in \mathbb{R}^N} \quad & { \lVert u \rVert }^2 + \kappa { \lVert \delta \rVert }^2 \\
	\textrm{s.t.} \quad & L_f \tilde{J}_1 (x) + L_g \tilde{J}_1 (x) u \le - \sigma_1(x) + \delta_1 \\
	\quad & \vdots \\
	\quad & L_f \tilde{J}_N (x) + L_g \tilde{J}_N (x) u \le - \sigma_N(x) + \delta_N \\
	\quad & K \delta \ge 0.
\end{aligned}
\end{equation}
Note that $L_f \tilde{J}_i (x)$ and $L_g \tilde{J}_i (x)$ are Lie derivatives where $L_f \tilde{J}_i (x) = \frac{ \partial \tilde{J}_i }{ \partial x} (x) f(x)$ and $L_g \tilde{J}_i (x) = \frac{ \partial \tilde{J}_i }{ \partial x} (x) g(x)$. 
Also note that since $\tilde{J}_1,\ldots,\tilde{J}_N$ are implemented with neural networks, $L_f \tilde{J}_i (x)$ and $L_g \tilde{J}_i (x)$ can be calculated using back-propagation.
Each $\tilde{J}_1,\ldots,\tilde{J}_N$ can be treated as a candidate Control Lyapunov Function (CLF) where picking an input that drives a $\tilde{J}_i$ toward $0$ is equivalent to making progress on that learned task.
$\sigma_1(x), \ldots, \sigma_N(x)$ are nonnegative scalar functions described in \cite{gennaroDars} that are designed to recover the optimal input for each learned cost-to-go function, assuming that each $\tilde{J}_i$ approximates an undiscounted cost functional of the form in (2) where $R_i(x) = I$. 
The controller in (3) also contains the slack variables $\delta_1, \ldots, \delta_N$ which can be used to prioritize certain tasks over others as specified by the matrix $K \in \mathbb{R}^{N \times N}$, which is required when the optimal input for each task cannot be selected

\section{Definitions for Relationships Between Learned Tasks}
We now formally define the notion of concurrently executing multiple learned control tasks.
Assuming that there are no additional constraints on the input, the ability to concurrently execute multiple learned value functions 
$\tilde{J}_1, \ldots, \tilde{J}_N$ using the min-norm controller in (3) depends on the geometric relationships between the Lie derivatives $L_g \tilde{J}_1, \ldots, L_g \tilde{J}_N$. 
Therefore, to classify the relationships between learned value functions and the ability to concurrently execute them, we provide Definitions 1-4 which are based on these Lie derivatives.
To start, the work in \cite{meAamas} defines notions of independence and orthogonality which we present again in Definitions 1 and 2.
\begin{deff}
	Consider a set of learned value functions $\tilde{J}_1, \ldots \tilde{J}_N$ of the form in (2).
	For each $x \in \mathcal{X}$ and $i \in \{1, \ldots, N\}$, we assume that $L_g \tilde{J}_i(x) = 0$ if and only if $\tilde{J}_i(x) = 0$.
	$\tilde{J}_1, \ldots, \tilde{J}_N$ are \textit{independent} to each other at state $x \in \mathcal{X}$ if the nonzero vectors within
	${L_g \tilde{J}_1(x)}^{\top}, \ldots, {L_g \tilde{J}_N(x)}^{\top}$ are linearly independent.
\end{deff}


\begin{deff}
	Consider a set of learned value functions $\tilde{J}_1, \ldots \tilde{J}_N$ of the form in (2).
	$\tilde{J}_1, \ldots, \tilde{J}_N$ are \textit{orthogonal} to each other at state $x$ if
	$\langle (L_g \tilde{J}_i (x) )^{\top}, (L_g \tilde{J}_j (x) )^{\top} \rangle = 0$ $\forall i, j \in \{1, \ldots, N \}$ where $i \ne j$.
\end{deff}

\begin{remark}
	Note that if value functions are orthogonal at $x$, then they are also independent at $x$.
\end{remark}

The work in \cite{meAamas} provided methods for training value functions to be independent to each other.
However, learned value functions $\tilde{J}_1, \ldots, \tilde{J}_N$ do not need to be independent at every state $x \in \mathcal{X}$ to be possible to concurrently execute together. 
For example, if for two value functions $\tilde{J}_i, \tilde{J}_j \in \{\tilde{J}_1, \ldots, \tilde{J}_N \}$, 
$L_g \tilde{J}_i$ and $L_g \tilde{J}_j$ were collinear, but pointed in the same direction, then it would still be possible to pick an input that made progress at both tasks.
Even in the case where $L_g \tilde{J}_i(x)$ and $L_g \tilde{J}_j(x)$ pointed in opposite directions, the values of $L_f \tilde{J}_i(x)$ and $L_f \tilde{J}_j (x)$ may still make it possible to make progress at both tasks.
This motivates the more general definitions of \textit{concurrent executability} and \textit{concurrent controllability} that we introduce in this paper, and refer to in the new results that we present in this work.
\begin{deff}
	Consider a set of learned value functions $\tilde{J}_1, \ldots \tilde{J}_N$ of the form in (2).
	$\tilde{J}_1, \ldots, \tilde{J}_N$ are \textit{concurrently executable} at $x \in \mathcal{X}$ if $\exists u \in \mathcal{U}$ such that $L_f \tilde{J}_i (x) + L_g \tilde{J}_i (x) u < 0$ for all $i \in \{1, \ldots, N \}$ where $\tilde{J}_i(x) > 0$. 
\end{deff}

\begin{deff}
	Consider a set of learned value functions $\tilde{J}_1, \ldots \tilde{J}_N$ of the form in (2).
	For each $x \in \mathcal{X}$ and $i \in \{1, \ldots, N\}$, we assume that $L_g \tilde{J}_i(x) = 0$ if and only if $\tilde{J}_i(x) = 0$.
	$\tilde{J}_1, \ldots \tilde{J}_N$ are \textit{concurrently controllable} at $x \in \mathcal{X}$ if there is no $P, Q \subseteq \{1, \ldots, N\}$ where $ \sum_{i \in P} L_g \tilde{J}_i(x) = - k \sum_{i \in Q}  L_g \tilde{J}_i (x)  $ and $\sum_{i \in P}  L_g \tilde{J}_i(x) \ne 0$ for some $k > 0$.
\end{deff}

\begin{remark}
	If a set of learned value functions $\tilde{J}_1, \ldots, \tilde{J}_N$ are \textit{concurrently controllable}, then they are \textit{concurrently executable}, assuming there are no additional constraints on the input.
	If a set of learned value functions $\tilde{J}_1, \ldots, \tilde{J}_N$ are \textit{independent}, then it is clear they are \textit{concurrently controllable} and thus also \textit{concurrently executable}.
\end{remark}

Now that we possess the formal language to denote what it means for a set of value functions to be concurrently executable, we proceed to the main theoretical results in the next section. Specifically, we present propositions which provide necessary and sufficient conditions for a set of value functions to be concurrently executable, as it is defined in Definition 3, within subsets of the state space.

\section{MAIN RESULTS}
Using the definitions presented in the previous section, we now establish conditions under which a set of value functions are concurrently executable within subsets of the state space. We first present necessary conditions that must be satisfied, followed by sufficient conditions that guarantee concurrent controllability, and hence concurrent executability. These results formalize intuitive requirements into precise necessary and sufficient conditions.
\subsection{Necessary Conditions for the Concurrent Executability of Learned Value Functions}
Consider a set of value functions $\tilde{J}_1, \ldots, \tilde{J}_N$ and a compact subset of the state space $A \subset \mathcal{X}$. 
For $\tilde{J}_1, \ldots, \tilde{J}_N$ to be concurrently executable everywhere in $A$, it may make intuitive sense that there must exist at least one common ``goal'' state $x \in A$ that is an equilibrium point where $\tilde{J}_1(x) = \cdots = \tilde{J}_N(x) = 0$. 
Furthermore, if $\tilde{J}_1, \ldots, \tilde{J}_N$ were concurrently executable everywhere in $A$, we would also expect that we never leave the sublevel set for any particular value function $\tilde{J}_i \in \{\tilde{J}_1, \ldots, \tilde{J}_N\}$ when executing each task simultaneously.
We prove that these intuitions are true and formalize them in Propositions 1 and 2 using what we first prove in Lemma 1.

\begin{lemm}
	Let $\tilde{J}_1, \ldots, \tilde{J}_N$ be a set of learned value functions as described in (2).
	Let $A \subset \mathcal{X}$ be a compact set. 
	Assume that $\forall x \in A$, $\exists u \in \mathbb{R}^m$ such that 
	$\dot{x} = f(x) + g(x) u \in T_A(x)$, the tangent cone to $A$ at $x$, and 
	$L_f \tilde{J}_i (x) + L_g \tilde{J}_i (x) u \le 0$ $\forall i \in \{1, \ldots, N \}$ where 
	$L_f \tilde{J}_i (x) + L_g \tilde{J}_i (x) u < 0$ when $\tilde{J}_i(x) > 0$.
	Then there exists some set of equilibrium states $G \subset A$ where $\forall x \in G$, $\tilde{J}_1(x) = \tilde{J}_2(x) = \cdots = \tilde{J}_N(x) = 0$ and $x$ is reachable from the set $A$ by staying within $A$. 
\end{lemm}
\begin{proof}
Consider a min-norm controller of the form:
	\begin{align*} \tag{4}
		\min_{u \in \mathbb{R}^m} \quad & { \lVert u \rVert }^2 \\
		\textrm{s.t.} \quad & L_f \tilde{J}_1 (x) + L_g \tilde{J}_1 (x) u \le - \sigma_1(x) \\
		\quad & \vdots \\
		\quad & L_f \tilde{J}_N (x) + L_g \tilde{J}_N (x) u \le - \sigma_N(x) \\
		\quad & f(x) + g(x) u \in T_A(x).
	\end{align*}
	Each of these constraints can be satisfied for all $x \in A$ under our assumptions.
	In other words, for any $x \in A$, there is some feasible input $u \in \mathbb{R}^m$ that stays in $A$ and also makes progress on each task $\tilde{J}_1, \ldots, \tilde{J}_N$.
	Let $\bar{V}(x) = \sum_{i = 1}^N \tilde{J}_i(x)$. 
	Let $\bar{f}(x) = f(x) + g(x) u$ be the dynamics of this autonomous system resulting from the specified controller.
	The set $A$ would be positively invariant and $L_{\bar{f}} \bar{V} (x) \le 0$ for each $x \in A$ along the resulting trajectory. Referring to \cite{khalil}, we can then apply Theorem 4.4 (LaSalle's theorem).
	We can see that $L_{\bar{f}} \bar{V}(x) = 0$ if and only if $\tilde{J}_1(x) = \tilde{J}_2(x) = \cdots = \tilde{J}_N(x) = 0$ and $x$ is an equilibrium state. 
	Let $G = \{ x \in A : \tilde{J}_1(x) = \cdots = \tilde{J}_N(x) = 0 \}$. 
	The proof in \cite{khalil} implicitly proves that the invariant set, which is $G$ in our case, is nonempty.
	Therefore, there exists some set of equilibrium states where $\tilde{J}_1(x) = \tilde{J}_2(x) = \cdots = \tilde{J}_N(x) = 0$ that is reachable from the set $A$ and by staying within the set $A$.
\end{proof}
\begin{prop}
	Let $\tilde{J}_1, \ldots, \tilde{J}_N$ be a set of learned value functions as described in (2).
	For any compact set $A \subset \mathcal{X}$, if there does not exist a set of common equilibrium states $G \subseteq A$ where $\forall x \in G$, $\tilde{J}_1(x) = \cdots = \tilde{J}_N(x) = 0$, then there exists some 
	$x \in A$ where either it is not possible to pick an input that causes the system to stay in $A$ or the tasks 
	are not concurrently executable in $A$. 
\end{prop}
\begin{proof}
	Lemma $1$ shows that if for every state in a compact set, $A \subset \mathcal{X}$, 
	there is an input that keeps the system in $A$ and concurrently makes progress on each task, then
	there exists some set of equilibrium states $G \subset A$ where $\forall x \in G$, 
	$\tilde{J}_1(x) = \cdots = \tilde{J}_N(x) = 0$. By the contrapositive of this, if such a set $G$ does
	not exist within any compact set, $A$, then for any compact set, $A$, there is some state
	$x \in A$ where it is either not possible to pick an input $u \in \mathcal{U}$ that keeps the system within
	$A$ or the value functions $\tilde{J}_1, \ldots, \tilde{J}_N$ are not concurrently executable at that state.
\end{proof}

\begin{prop}
	Let $\tilde{J}_1, \ldots, \tilde{J}_N$ be a set of learned value functions as described in (2).
	Let $A \subset \mathcal{X}$ be a compact set of states.
	Let $G$ be the common set of equilibrium states where $\forall x \in G$, 
	$\tilde{J}_1(x) = \cdots = \tilde{J}_N(x) = 0$.
	For $i \in \{1, \ldots, N\}$ and $c \ge 0$, let $S_{i , c} := \{ x \in A : \tilde{J}_i(x) \le c \}$.
	If there exists some $i \in \{1, \ldots, N\}$ and $c \ge 0$ where no state in $G$ is reachable from $S_{i, c}$
	without leaving $S_{i, c}$, then either $\tilde{J}_1, \ldots, \tilde{J}_N$ are not concurrently executable 
	at every state in $A$ or there exists some state $x \in A$ where it is not possible to pick an input that keeps the system in $A$. 
\end{prop}
\begin{proof}
	We can see this by applying the contrapositive of Lemma $1$ to the compact set $S_{i, c}$ and value functions $\tilde{J}_1, \ldots, \tilde{J}_{i - 1}, \tilde{J}_{i + 1}, \tilde{J}_N$ for some $i \in \{1, \ldots, N \}$ and $c \ge 0$.
	If there does not exist a set of equilibrium states $\bar{G} \subset S_{i, c}$ where $\forall x \in \bar{G}$, $\tilde{J}_1(x) = \cdots = \tilde{J}_{i - 1}(x) = \tilde{J}_{i + 1}(x) = \cdots = \tilde{J}_N(x) = 0$, or if $\bar{G}$ exists but it is possible to reach $\bar{G}$ without leaving
	$S_{i, c}$, then there is some $x \in S_{i, c}$ where it is either not possible to execute each of the tasks $\tilde{J}_1, \ldots, \tilde{J}_{i - 1}, \tilde{J}_{i + 1}, \tilde{J}_N$ or it is not possible
	to pick an input that keeps the state in $S_{i, c}$, meaning that we either cannot make progress on the task $\tilde{J}_i$ or we exit the larger set $A$. In either case, this means that
	$\tilde{J}_1, \ldots, \tilde{J}_N$ are not concurrently executable at that state $x \in S_{i, c} \subseteq A$, or it is not possible to pick an input that keeps the system within the set $A$. 
\end{proof}


\subsection{Existence of a Concurrently Executable Set of States}
If there does exist at least one common ``goal'' equilibrium state $x \in \mathcal{X}$ where $\tilde{J}_1(x) = \cdots = \tilde{J}_N(x) = 0$, it would intuitively seem that $\tilde{J}_1, \ldots, \tilde{J}_N$ should at least be concurrently executable when close enough to such a ``goal'' state. We prove in Proposition 3 that this intuition is also true.

\begin{prop}
	Let $\tilde{J}_1, \ldots, \tilde{J}_N$ be a set of learned value functions of the form in (2) that are all smooth and not identically equal to $0$.
	Assume there exists some non-empty set $G \subset \mathcal{X}$ where $\forall x \in G$, $x$ is an equilibrium point and $\tilde{J}_1(x) = \cdots = \tilde{J}_N(x) = 0$. 
	Referring to the system dynamics in (1), assume that $g(x)$ has full row rank.
	Then there exists some set $F$, where $G \subset F \subseteq \mathcal{X}$, and for every $x \in F$, $\tilde{J}_1, \ldots, \tilde{J}_N$ is concurrently controllable, and therefore also concurrently executable, at the state $x$. 
\end{prop}
\begin{proof}
	If $g(x)$ has full row rank, then $\tilde{J}_1, \ldots, \tilde{J}_N$ being concurrently controllable or not relies solely on the gradients of each value function. We can see that if there exists some $P, Q \subset \{1, \ldots, N\}$, $x \in \mathcal{X}$ and $k > 0$ such that $ \sum_{i \in P} L_g \tilde{J}_i(x) = - k \sum_{i \in Q}  L_g \tilde{J}_i (x)$  where neither summation of Lie derivatives are zero, then it must be the case that:
	\begin{align*} \tag{5}
		\sum_{i \in P} \nabla_x \tilde{J}_i &= - k \sum_{i \in Q} \nabla_x \tilde{J}_i.
	\end{align*}
	Now we must see if there exists some set $F \supset G$ such that no state exists within $F$ where equation (5) is true.\par
Let $B_{\delta} := \{ x \in \mathcal{X} : \exists g \in G, \lVert x - g \rVert < \delta \}$.
For a suitably small $\delta > 0$, let $a \in B_{\delta}$ and $x \in G$ such that:
\begin{align*} \tag{6}
	\tilde{J}_1 (x) = 0 &= \tilde{J}_1(a) + { \nabla_a \tilde{J}_1 }^{\top} \left( x - a \right) + \mathcal{O}(x, a) \\
	\vdots \\
	\tilde{J}_N (x) = 0 &= \tilde{J}_N(a) + { \nabla_a \tilde{J}_N }^{\top} \left( x - a \right) + \mathcal{O}(x, a)
\end{align*}
where $\mathcal{O}(x, a)$ represents the higher order terms in each Taylor expansion.
As the $\delta$ approaches $0$, it is reasonable to ignore the higher order terms to get:
\begin{align*} \tag{7}
	\tilde{J}_1 (x) = 0 &= \tilde{J}_1(a) + { \nabla_a \tilde{J}_1 }^{\top} \left( x - a \right)\\
	\vdots \\
	\tilde{J}_N (x) = 0 &= \tilde{J}_N(a) + { \nabla_a \tilde{J}_N }^{\top} \left( x - a \right).
\end{align*}
If $\nabla_a \tilde{J}_i = 0$ for each $i \in \{1, \ldots, N\}$, then $\tilde{J}_1, \ldots, \tilde{J}_N$ are clearly, by Definition 4, concurrently controllable at $a$. 
Let $P, Q \subset \{1, \ldots, N\}$. 
Let $p = \sum_{i \in P} \tilde{J}_i (a)$ and let $q = \sum_{i \in Q} \tilde{J}_i (a)$.
Without loss of generality, assume that $p, q > 0$.
We can then see that:
$$
	0 = p + \langle \sum_{i \in P} \nabla_a \tilde{J}_i, x - a \rangle = q + \langle \sum_{i \in Q} \nabla_a \tilde{J}_i, x - a \rangle. \eqno{(8)}
$$
	Let $k = \frac{p}{q}$. Multiplying the rightmost expression in (8) by $k$, we see that:
\begin{align*} \tag{9}
	\langle \sum_{i \in P} \nabla_a \tilde{J}_i, x - a \rangle = k \langle \sum_{i \in Q} \nabla_a \tilde{J}_i, x - a \rangle.
\end{align*}
This means that $P, Q$ cannot be such that $\sum_{i \in P} \nabla_a \tilde{J}_i$ and 
$\sum_{i \in Q} \nabla_a \tilde{J}_i$ point in opposite directions
which means that $\tilde{J}_1, \ldots, \tilde{J}_N$ must be concurrently controllable at $a$.
Therefore, for a small enough value of $\delta$, $\tilde{J}_1, \ldots, \tilde{J}_N$ are concurrently controllable at 
the states within $B_{\delta}$. Since none of $\tilde{J}_1, \ldots, \tilde{J}_N$ are identically equal to $0$, 
and each of $\tilde{J}_1, \ldots, \tilde{J}_N$ are continuous, it must be the case that $B_{\delta}$ is non-empty.
\end{proof}
In Proposition 3, we show that when a set of value functions, $\tilde{J}_1, \ldots, \tilde{J}_N$, have a common ``goal'' state, an equilibrium state where each value function is equal to $0$, there is at least some non-empty set 
$F \subset \mathcal{X}$ where each of the value functions are concurrently executable.

\begin{remark}
In the previous proof, we assume that $g(x)$ has full row rank. 
This may seem like a strong assumption, but for a system to have the necessary redundancy to concurrently execute multiple control tasks, it must be overactuated, making the assumption not unreasonable in this context.
\end{remark}

\subsection{Recovering the Optimal Input for Discounted Value Functions}
The results in subsections A and B are concerned with the concurrent executability of learned value functions using the min-norm controller in (3). These results are agnostic to the chosen expressions for $\sigma_i$ in (3).
However, to apply those results in practice, the min-norm controller should recover inputs consistent with the HJB equation satisfied by the learned value functions. This subsection addresses that requirement.\par
For the controller in (3), the work in \cite{gennaroDars} proposes an expression on the right-hand side of each constraint of the form:
$$
\sigma_i(x) = \sqrt{ \left( L_f \tilde{J}_i(x) \right)^2 + q(x) L_g \tilde{J}_i(x) {L_g \tilde{J}_i(x)}^{\top} }. \eqno{(10)}
$$
This recovers the optimal input in the case where $\tilde{J}_i$ approximates a value function of the form in (2)
where it is undiscounted, meaning that $\beta = 0$, and $R(x) = I$.
However, it is well known that RL algorithms tend to converge with fewer iterations when a discount factor is used, which corresponds to some $\beta > 0$ in (3).
In this paper, we propose a more general expression to be used:
$$
\bar{\sigma}_i(x) = \sqrt{ \left( L_f \tilde{J}_i(x) \right)^2 + \left( q(x) - \beta \tilde{J}_i(x) \right) L_g \tilde{J}_i(x) {L_g \tilde{J}_i(x)}^{\top} }
$$
which works for both the discounted and undiscounted case.
When $R(x) = I$ in (3), the optimal input is $u^{\star}(x) = - \frac{1}{2} {g(x)}^{\top} \nabla_x \tilde{J}_i$.
Substituting this into the Hamilton-Jacobi-Bellman (HJB) equation results in:
\begin{align*}
& \frac{1}{4} \left( L_g \tilde{J}_i (x) \right) \left( L_g \tilde{J}_i (x) \right)^{\top} - L_f \tilde{J}_i (x) 
&= q(x) - \beta \tilde{J}_i (x).
\end{align*}
If $\tilde{J}_i$ were the only value function executed by the min-norm controller in (3), and $\bar{\sigma}_i(x)$ was used in place of $\sigma_i$, the constraint would be active and $u^{\star}(x) = - \frac{1}{2} {g(x)}^{\top} \nabla_x \tilde{J}_i$ would be the resulting input, ignoring the effect of slack variables, regardless of whether or not $\tilde{J}_i$ was trained with a discount factor.
With this modification, the min-norm controller can now accommodate the value functions that are trained with a discount factor -- something that is commonly done within Reinforcement Learning.

\section{SIMULATION RESULTS}
In this section, we provide illustrative examples demonstrating the propositions presented in the previous section with scenarios involving a planar mobile robot.
Let $\mathcal{X} = \mathbb{R}^2$ be the robot position and the input space $\mathcal{U} = \mathbb{R}^2$, so that $\dot{x}(t) = u(t)$.
By picking a two-dimensional example, we are able to clearly illustrate the implications of Propositions 1, 2 and 3 which provide insight into when it is possible for trained value functions to be concurrently executable. 
\par
\begin{example}
First we will demonstrate the implications of Proposition 1, which tells us that if for a set of value functions $\tilde{J}_1, \ldots, \tilde{J}_N$ there is no $x \in \mathcal{X}$ where $\tilde{J}_1(x) = \cdots = \tilde{J}_N(x) = 0$, then there is no compact set $A \subset \mathcal{X}$ where it is possible to concurrently execute each task without leaving the set $A$. Consider the case where we train two value functions -- each encoding the task of moving to a distinct point in the state space. 
$\tilde{J}_1$ is trained to go to the point ${\begin{bmatrix} -1.5 & 1.5 \end{bmatrix}}^{\top}$
with an instantaneous cost of $q_1(x) + { \lVert u \rVert }^2$ and
$\tilde{J}_2$ is trained to go to the point ${\begin{bmatrix} 1.5 & -1.5 \end{bmatrix}}^{\top}$
with an instantaneous cost of $q_2(x) + { \lVert u \rVert }^2$.
$q_1(x)$ and $q_2(x)$ are the distances from ${\begin{bmatrix} -1.5 & 1.5 \end{bmatrix}}^{\top}$
and ${\begin{bmatrix} 1.5 & -1.5 \end{bmatrix}}^{\top}$, respectively, multiplied by $5$.
Heatmaps of the two trained value functions can be found in Fig. 1.
To numerically evaluate the relationship between the Lie derivatives, which simplify to the gradients in the single integrator case, between  $\tilde{J}_1$ and $\tilde{J}_2$, we plotted a heatmap of the angle between the gradients of $\tilde{J}_1$ and $\tilde{J}_2$ in Fig. 2.
Points that are closer to red, in either heatmap, are points where the value functions are not concurrently executable, as the angle between the gradients are around $\pi$, meaning that they point in opposite directions.
Points that are closer to green are points where the angle between the two gradients are close to $0$.
As expected, we can see a line of states where $\tilde{J}_1$ and $\tilde{J}_2$ are not concurrently executable.
We claim that it would not be possible to find a compact set $A \subset \mathcal{X}$ where it would be possible
to concurrently execute both $\tilde{J}_1$ and $\tilde{J}_2$ without leaving $A$.
A trajectory resulting from using the min-norm controller from (3) with $\tilde{J}_1$ and $\tilde{J}_2$ encoded as constraints is also shown in Fig. 2, illustrating the effect of the states that are not concurrently executable as the trajectory gets stuck.\par
\begin{figure}[]
	\centering
	\begin{minipage}[c]{0.28\linewidth}
		\centering
		\includegraphics[width=0.9\linewidth]{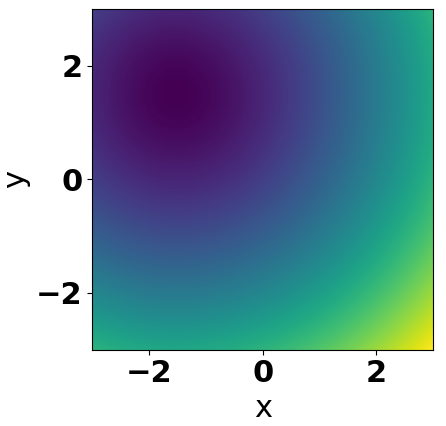}
	\end{minipage}
	\begin{minipage}[c]{0.28\linewidth}
		\centering
		\includegraphics[width=0.9\linewidth]{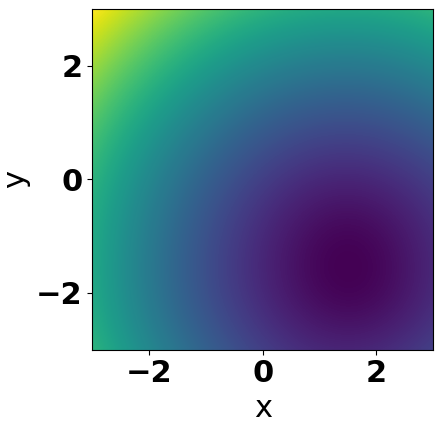}
	\end{minipage}%
	\begin{minipage}[c]{0.31\linewidth}
		\centering
		\includegraphics[width=0.9\linewidth]{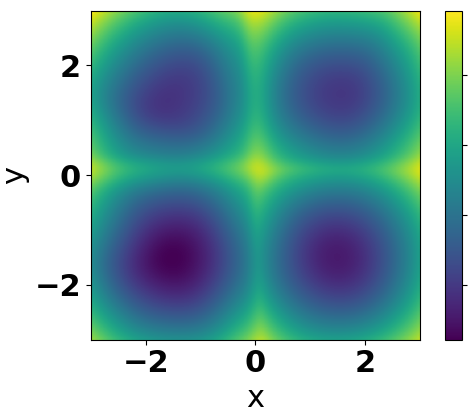}
	\end{minipage}%
	\caption{Heatmaps of $\tilde{J}_1$, $\tilde{J}_2$, and $\tilde{J}_3$, respectively. Points closer to yellow represent larger values of the learned cost-to-go functions $\tilde{J}_1$, $\tilde{J}_2$, and $\tilde{J}_3$.}
\end{figure}
\end{example}
\begin{example}
Next, we will demonstrate the implications of Propositions 2 and 3.
$\tilde{J}_3$ is trained to be able to go to the nearest of four points among
$ {\begin{bmatrix} -1.5 & 1.5 \end{bmatrix}}^{\top}$, $ {\begin{bmatrix} 1.5 & 1.5 \end{bmatrix}}^{\top} $, $ {\begin{bmatrix} -1.5 & -1.5 \end{bmatrix}}^{\top}$ and $ {\begin{bmatrix} 1.5 & -1.5 \end{bmatrix}}^{\top} $.
It is trained using an instantaneous cost of $q_3(x) + { \lVert u \rVert }^2$ where $q_3(x)$ is the distance to the nearest of the four points multiplied by $5$.
A heatmap of $\tilde{J}_3$ is plotted in Fig. 1.\par
Recall that we previously trained $\tilde{J}_2$ to go to the point ${\begin{bmatrix} 1.5 & -1.5 \end{bmatrix}}^{\top}$.
To analyze the relationship between $\tilde{J}_2$ and $\tilde{J}_3$, we again plotted a heatmap in Fig. 2 with the same color code as before.
Since $\tilde{J}_2$ and $\tilde{J}_3$ share a ``goal'' state, we would expect, based on Proposition 3, that there exists some subset of the state space around the point ${\begin{bmatrix} 1.5 & -1.5 \end{bmatrix}}^{\top}$ where both $\tilde{J}_2$ and $\tilde{J}_3$ are concurrently executable.
We can see that this is the case and it would be possible to draw a ball around the point where both
$\tilde{J}_2$ and $\tilde{J}_3$ are concurrently executable at every state in the ball.
However, as we expand the ball, we can see that it would eventually encompass states where $\tilde{J}_2$ and $\tilde{J}_3$ are not concurrently executable.
Near these states, we can see that it would not be possible to make progress on one of the tasks without leaving the sublevel set of the other value function.
$\tilde{J}_2$ and $\tilde{J}_3$ not being concurrently executable at these states is an implication of Proposition 2.
Two different trajectories are plotted in Fig. 2. We can see that one trajectory gets stuck as it ends up at a state where $\tilde{J}_2$ and $\tilde{J}_3$ are not concurrently executable while the other trajectory which started close enough to a goal state successfully reaches the goal state.
\par
\begin{figure}[]
	\centering
	\begin{minipage}[c]{0.40\linewidth}
		\centering
		\includegraphics[width=0.9\linewidth]{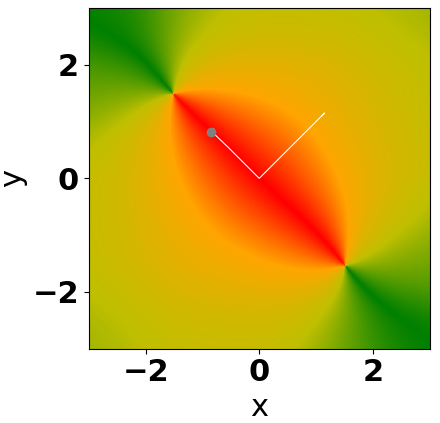}
	\end{minipage}
	\begin{minipage}[c]{0.48\linewidth}
		\centering
		\includegraphics[width=0.9\linewidth]{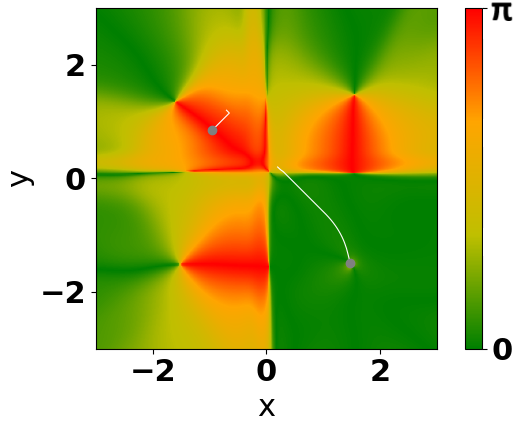}
	\end{minipage}%
	\caption{
	Heatmap of angles between the gradients of $\tilde{J}_1$ and $\tilde{J}_2$ on the left corresponding to Example 1.
	Heatmap of angles between the gradients of $\tilde{J}_2$ and $\tilde{J}_3$ on the right corresponding to Example 2.
	Points that are closer to red are points where the two value functions in the scenario are not concurrently executable.
	Both images contain trajectories, in white, of the min-norm controller executing both tasks in each scenario from different initial states, with the gray circle being the final position of each trajectory.
	In Example 1, on the left, the trajectory gets stuck as it ends up in a state that is not concurrently executable.
	In Example 2, on the right, one trajectory gets stuck while the other that starts close enough to the concurrently executable region reaches one of the four goal points.
	}
\end{figure}
\end{example}
These examples may point toward guidelines for practitioners aiming to use the framework developed in \cite{gennaroDars}, \cite{meAamas}. When a set of value functions are not concurrently executable, state cost functions may need to be redesigned so that at least one common goal state exists where each of the learned value functions would equal $0$.
When there still exist parts of the state space that are not concurrently executable, some tasks may need to be prioritized through the use of the slack variables in (3).

\section{CONCLUSION}
In this paper, we established necessary and sufficient conditions for the concurrent execution of learned value functions when utilizing the constraint-based framework of \cite{gennaroDars}, \cite{meAamas}.
The presented theorems rigorously characterize when control tasks learned using RL can be made concurrently executable, when they might inherently be so and when concurrent execution of tasks is impossible.
Additionally, we present a modification to the min-norm controller in order to account for value functions trained with a discount factor, making the framework more compatible with typical RL practices.
The theorems were demonstrated using examples involving a planar mobile robot, illustrating the effects of the necessary and sufficient conditions in attempting to combine learned value functions.
While the results in this paper directly apply to the constraint-based framework of \cite{gennaroDars}, \cite{meAamas}, the results may have broader implications in the integration of RL and control-theoretic tools.
For example, these results may provide a theoretical basis for the rigorous development or analysis of multi-objective RL methods.
\addtolength{\textheight}{-12cm}   







\bibliographystyle{ieeetr} 
\bibliography{root}

@InProceedings{gennaroDars,
author="Notomista, Gennaro",
editor="Bourgeois, Julien
and Paik, Jamie
and Piranda, Beno{\^i}t
and Werfel, Justin
and Hauert, Sabine
and Pierson, Alyssa
and Hamann, Heiko
and Lam, Tin Lun
and Matsuno, Fumitoshi
and Mehr, Negar
and Makhoul, Abdallah",
title="A Constrained-Optimization Approach to the Execution of Prioritized Stacks of Learned Multi-robot Tasks",
booktitle="Distributed Autonomous Robotic Systems",
year="2024",
publisher="Springer Nature Switzerland",
address="Cham",
pages="479--493",
isbn="978-3-031-51497-5"
}

@ARTICLE{actorCriticSurvey,
author={Grondman, Ivo and Busoniu, Lucian and Lopes, Gabriel A. D. and Babuska, Robert},
journal={IEEE Transactions on Systems, Man, and Cybernetics, Part C (Applications and Reviews)}, 
title={A Survey of Actor-Critic Reinforcement Learning: Standard and Natural Policy Gradients}, 
year={2012},
volume={42},
number={6},
pages={1291-1307},
doi={10.1109/TSMCC.2012.2218595}}

@InProceedings{cfvi,
title = 	 {Value Iteration in Continuous Actions, States and Time},
author =       {Lutter, Michael and Mannor, Shie and Peters, Jan and Fox, Dieter and Garg, Animesh},
booktitle = 	 {Proceedings of the 38th International Conference on Machine Learning},
pages = 	 {7224--7234},
year = 	 {2021},
editor = 	 {Meila, Marina and Zhang, Tong},
volume = 	 {139},
series = 	 {Proceedings of Machine Learning Research},
month = 	 {18--24 Jul},
publisher =    {PMLR},
url = 	 {https://proceedings.mlr.press/v139/lutter21a.html}
}

@misc{esb,
      title={Extended Set-based Tasks for Multi-task Execution and Prioritization},
      author={Gennaro Notomista and Mario Selvaggio and Francesca Pagano and María Santos and Siddharth Mayya and Vincenzo Lippiello and Cristian Secchi},
      year={2025},
      eprint={2310.16189},
      archivePrefix={arXiv},
      primaryClass={eess.SY},
      url={https://arxiv.org/abs/2310.16189},
}

@ARTICLE{jacobianTasks,
author={Antonelli, Gianluca},
journal={IEEE Transactions on Robotics}, 
title={Stability Analysis for Prioritized Closed-Loop Inverse Kinematic Algorithms for Redundant Robotic Systems}, 
year={2009},
volume={25},
number={5},
pages={985-994},
doi={10.1109/TRO.2009.2017135}}

@book{bertsekas,
  title={Reinforcement Learning and Optimal Control},
  author={Bertsekas, D.},
  isbn={9781886529397},
  series={Athena Scientific optimization and computation series},
  url={https://books.google.ca/books?id=ZlBIyQEACAAJ},
  year={2019},
  publisher={Athena Scientific}
}

@book{redundancyBook,
	author = {Milutinovic, Dejan and Rosen, Jacob},
	title = {Redundancy in Robot Manipulators and Multi-Robot Systems},
	year = {2012},
	isbn = {3642339700},
	publisher = {Springer Publishing Company, Incorporated}
}

@article{robotManipulatorsRedundancy1,
	author = {Yoshihiko Nakamura and Hideo Hanafusa and Tsuneo Yoshikawa},
	title ={Task-Priority Based Redundancy Control of Robot Manipulators},

	journal = {The International Journal of Robotics Research},
	volume = {6},
	number = {2},
	pages = {3-15},
	year = {1987},
	doi = {10.1177/027836498700600201},

	URL = { 
		        https://doi.org/10.1177/027836498700600201
	},
	eprint = { 
		        https://doi.org/10.1177/027836498700600201
	}
}

@ARTICLE{multiRobotRedundancy1,
	  author={Roehr, Thomas M.},
	    journal={IEEE Transactions on Robotics}, 
	      title={Active Exploitation of Redundancies in Reconfigurable Multirobot Systems}, 
	        year={2022},
		  volume={38},
		    number={1},
		      pages={180-196},
			  doi={10.1109/TRO.2021.3118284}}

@INPROCEEDINGS{gennaroEcc,
	  author={Notomista, Gennaro and Mayya, Siddharth and Hutchinson, Seth and Egerstedt, Magnus},
	    booktitle={2019 18th European Control Conference (ECC)}, 
	      title={An Optimal Task Allocation Strategy for Heterogeneous Multi-Robot Systems}, 
	        year={2019},
		  volume={},
		    number={},
		      pages={2071-2076},
			  doi={10.23919/ECC.2019.8795895}}

@book{khalil,
      author        = "Khalil, Hassan K",
      title         = "{Nonlinear systems; 3rd ed.}",
      publisher     = "Prentice-Hall",
      address       = "Upper Saddle River, NJ",
      year          = "2002",
      url           = "https://cds.cern.ch/record/1173048",
}

@inproceedings{meAamas,
author = {Tahmid, Sheikh A. and Notomista, Gennaro},
title = {Value Iteration for Learning Concurrently Executable Robotic Control Tasks},
year = {2025},
isbn = {9798400714269},
publisher = {International Foundation for Autonomous Agents and Multiagent Systems},
address = {Richland, SC},
booktitle = {Proceedings of the 24th International Conference on Autonomous Agents and Multiagent Systems},
pages = {2006–2014},
numpages = {9},
location = {Detroit, MI, USA},
series = {AAMAS '25}
}

@InProceedings{sac,
  title = 	 {Soft Actor-Critic: Off-Policy Maximum Entropy Deep Reinforcement Learning with a Stochastic Actor},
  author =       {Haarnoja, Tuomas and Zhou, Aurick and Abbeel, Pieter and Levine, Sergey},
  booktitle = 	 {Proceedings of the 35th International Conference on Machine Learning},
  pages = 	 {1861--1870},
  year = 	 {2018},
  editor = 	 {Dy, Jennifer and Krause, Andreas},
  volume = 	 {80},
  series = 	 {Proceedings of Machine Learning Research},
  month = 	 {10--15 Jul},
  publisher =    {PMLR},
  url = 	 {https://proceedings.mlr.press/v80/haarnoja18b.html}
}

@ARTICLE{khatib,
  author={Khatib, O.},
  journal={IEEE Journal on Robotics and Automation}, 
  title={A unified approach for motion and force control of robot manipulators: The operational space formulation}, 
  year={1987},
  volume={3},
  number={1},
  pages={43-53},
  doi={10.1109/JRA.1987.1087068}}

\end{document}